\documentclass[11pt]{article}
\usepackage{longtable}
\usepackage{caption}
\usepackage{multirow,multicol}
\usepackage{epsf, graphicx}
\usepackage{latexsym,amsfonts,amsbsy,amssymb}
\usepackage{amsmath,amsthm}
\usepackage{float}
\usepackage[colorlinks,linkcolor=blue,anchorcolor=blue,citecolor=blue]{hyperref}

\usepackage[linesnumbered,lined,boxed,commentsnumbered]{algorithm2e}
\usepackage{threeparttable}
\usepackage{cleveref}
\usepackage{indentfirst}
\usepackage{bm}
\allowdisplaybreaks[4]
\makeatletter

\DeclareMathOperator{\GRS}{GRS}
\DeclareMathOperator{\TGRS}{TGRS}
\@addtoreset{equation}{section} \makeatother
\newtheorem{Theorem}{Theorem}[section]

\newtheorem{Remark}{Remark}[section]
\newtheorem{Definition}{Definition}[section]
\newtheorem{Proposition}{Proposition}[section]
\newtheorem{Example}{Example}[section]
\makeatletter \@addtoreset{equation}{section} \makeatother
\begin{document}
	
	\title{Improved Decoding Algorithms for MDS and Almost-MDS Codes from Twisted GRS Codes \let\thefootnote\relax\footnote{E-mail addresses: wanggdmath@163.com (G. Wang), hwliu@ccnu.edu.cn (H. Liu), luojinquan@ccnu.edu.cn (J. Luo)
	}}
	\author{Guodong Wang, Hongwei Liu, Jinquan Luo}
	\date{\small
		School of Mathematics and Statistics,
		Central China Normal University,
		Wuhan,  430079, China\\
	}
	\maketitle
	{\noindent\small{\bf Abstract:}
In this paper, firstly, we study decoding of a general class of twisted generalized Reed-Solomon (TGRS) codes and provide a precise characterization of the key equation for TGRS codes and propose a decoding algorithm.
Secondly, we further study decoding of almost-MDS TGRS codes and provide a decoding algorithm.
These two decoding algorithms are more efficient in terms of performance compared with the decoding algorithms presented in [Sun et al., IEEE-TIT, 2024] and [Sui et al., IEEE-TIT, 2023] respectively.
Moreover, these two optimized decoding algorithms can be applied to the decoding of a more general class of twisted Goppa codes.
}
	\vspace{1ex}

	

	\section[Introduction]{Introduction}
    A linear code with parameters $[n, k, d]$ is referred to as a maximum distance separable (MDS) code if it meets the Singleton bound,
i.e., $d = n - k + 1$. MDS codes, due to their excellent properties,
have garnered extensive attention. When $d = n - k$,
the linear code is called almost-MDS. Various types of MDS codes exist,
and numerous methods for constructing these codes have been proposed \cite{Ezerman, Fang, Niu, Sok, Wu2}.

The generalized Reed-Solomon (GRS) codes stand out as a crucial class of MDS codes,
distinguished by its remarkable error correction capability, streamlined algebraic structure,
and efficient decoding algorithms.
Goppa codes, which are subfield subcodes of GRS codes introduced by Goppa in \cite{Goppa1, Goppa2},
have garnered significant attention from scholars due to their application in the McEliece and Niederreiter cryptosystems \cite{Berger, Li, Sui3}.

Niederreiter was the first researcher to suggest a public-key system using GRS codes \cite{Niederreiter},
but this system later turned out to be susceptible to the Sidelnikov-Shestakov attack \cite{Sidelnikov}.
Subsequently, Beelen et al. introduced twisted Reed-Solomon (TRS) codes in \cite{Beelen1},
presenting novel general constructions of MDS codes that are not equivalent to GRS codes.
In \cite{Beelen2}, Beelen et al. further investigated the structure of TRS codes and proposed using TRS codes as a substitute for Goppa codes in McEliece cryptosystems.
Thereafter, Lavauzelle et al. developed an efficient key recovery algorithm specifically for cryptosystems based on TRS codes \cite{Lavauzelle}.

Following this line, the research has provided insights into their dual, self-dual, linear complementary dual (LCD),
and their hulls (the intersections of these codes and their duals), as detailed in \cite{Guo, Huang1, Huang2, Sui1, Sui2, Sui4, Wu3}.
More recently, multiple twists GRS codes have been studied in \cite{Hu, Meena, Zhao}.

On the other hand, effective decoding algorithms are pivotal in the study of error-correcting codes.
Various methods for decoding GRS codes were studied,
including the Peterson-Gorenstein-Zierler Algorithm \cite{Peterson}, the Berlekamp-Massey Algorithm \cite{Berlekamp2}, and the Sugiyama Algorithm \cite{Sugiyama}.
The Sugiyama Algorithm leverages the Euclid's Algorithm for polynomials in a straightforward and potent way.
In \cite{Sui3}, Sui et al. explored generalized Goppa codes, which were applicable to the Niederreiter public key cryptosystem,
and introduced an efficient decoding algorithm for twisted Goppa codes based on the extended Euclid's Algorithm.
However, this algorithm could only correct $\lfloor \frac{t-1}{2} \rfloor$ errors when the minimum distance $d$ of the  Goppa code is at least $t+1$,
where $t$ is the degree of the Goppa polynomial $g(x)$ and $\lfloor a \rfloor$ denotes the greatest integer $\le a$.
Based on the work in \cite{Sui3}, Sun et al. in \cite{Sun1} improved the results.
They provided decoding algorithms which can correct $\lfloor \frac{t}{2} \rfloor$ errors for two classes of MDS TGRS codes and a class of twisted Goppa codes,
where the minimum distance $d$ is at least $t+1$ and $t$ is even.

The key problem of decoding a TGRS code is to solve the following key equation
\begin{align*}
    S(x)\sigma(x) \equiv \tau(x) \pmod{g(x)}
\end{align*}
for given $S(x)$ and $g(x)$, where the degree of $\sigma(x)$ is equal to the number of errors
and $\deg \tau(x) \le \deg \sigma(x)$.
Sun et al. provided the key equation for decoding MDS TGRS codes and presented the corresponding decoding algorithm in \cite{Sun1}.
The decoding processes for two types of MDS TGRS codes are discussed,
with respective parity-check matrices given as follows:
\[
H_1 =\begin{pmatrix}
 v_1(1+\eta\alpha_1^t) & \cdots & v_n(1+\eta\alpha_n^t) \\
v_1\alpha_1 & \cdots & v_n\alpha_n \\
\vdots & & \vdots \\
v_1\alpha_1^{t-2} & \cdots & v_n\alpha_n^{t-2} \\
v_1\alpha_1^{t-1} & \cdots & v_n\alpha_n^{t-1}
\end{pmatrix}
\]
and
\[
H_2 =\begin{pmatrix}
v_1 & \cdots & v_n \\
v_1\alpha_1 & \cdots & v_n\alpha_n \\
\vdots & & \vdots \\
v_1\alpha_1^{t-2} & \cdots & v_n\alpha_n^{t-2} \\
v_1(\alpha_1^{t-1}+\eta\alpha_1^t) & \cdots & v_n(\alpha_n^{t-1}+\eta\alpha_n^t)
\end{pmatrix}.
\]
These two types of MDS TGRS codes have generator matrices which are given by:
\[
\small{
G_1 = \begin{pmatrix}
    w_1 & w_2 & \cdots & w_n \\
    \vdots & \vdots & \ddots & \vdots \\
    w_1 \alpha_1^{n-t-2} & w_2 \alpha_2^{n-t-2} & \cdots & w_n \alpha_n^{n-t-2} \\
    w_1( \alpha_1^{n-t-1} + b_1\alpha_1^{-1}) & w_2( \alpha_2^{n-t-1} + b_1\alpha_2^{-1}) & \cdots & w_n( \alpha_n^{n-t-1} + b_1\alpha_n^{-1})
\end{pmatrix}
}
\]
and
\[
\small{
G_2=\begin{pmatrix}
    w_1 & w_2 & \cdots & w_n \\
    w_1 \alpha_1 & w_2 \alpha_2 & \cdots & w_n \alpha_n \\
    \vdots & \vdots & \ddots & \vdots \\
    w_1 \alpha_1^{n-t-2} & w_2 \alpha_2^{n-t-2} & \cdots & w_n \alpha_n^{n-t-2} \\
    w_1(b_2\alpha_1^{n-t-1} + \alpha_1^{n-t}) & w_2(b_2\alpha_2^{n-t-1} + \alpha_2^{n-t}) & \cdots & w_n(b_2\alpha_n^{n-t-1} + \alpha_n^{n-t})
\end{pmatrix},
}
\]
where $$b_1 = -\frac{\eta\sum\limits_{i=1}^n u_i \alpha_i^{n-1} +  \sum\limits_{i=1}^n u_i \alpha_i^{n-t-1}}{\sum\limits_{i=1}^n u_i \alpha_i^{-1}}(t > 1),
b_2 = -\frac{\sum\limits_{i=1}^n u_i \alpha_i^{n-1} + \eta \sum\limits_{i=1}^n u_i \alpha_i^n}{\eta \sum_{i=1}^n u_i \alpha_i^{n-1}}, w_i=\frac{u_i}{v_i}$$ and $$u_i^{-1}= \prod_{j=1,j\ne i}^n(\alpha_i-\alpha_j), 1\le i \le n.$$
According to Definitions 2.2 and 2.3 (in Section~2),
these two types of TGRS codes are subclasses of the TGRS codes defined in this paper.

In this paper, we study the decoding of a general class of TGRS and provide a more precise characterization of the key equation for TGRS codes.
This characterization aids in optimizing the algorithm presented in \cite{Sun1}, and we have also proposed the optimized decoding algorithm.
We further study the decoding of almost-MDS TGRS codes and provide the optimized decoding algorithm which is more efficient than the decoding algorithm presented in \cite{Sui3} in terms of performance.
Moreover, these two optimized decoding algorithms can be applied to the decoding of a general class of twisted Goppa codes.

This paper is organized as follows.
In Section 2, we introduce some basic notations and definitions of TGRS codes.
In Section 3, we present parity-check matrices of the TGRS codes defined in this paper.
In Section 4, we discuss the decoding of a class of MDS or almost-MDS TGRS codes.
In Section 5, we utilize extended Euclid's Algorithm to provide decoding algorithms for TGRS codes in both MDS and almost-MDS scenarios.
In Section 6, we define a larger class of twisted Goppa codes, and their decoding can reuse the decoding algorithms for the TGRS codes.
Finally, Section 7 concludes this paper. And the performance comparison results between our algorithm and existing algorithms are presented in Table 2.

\section{Preliminaries}
 Let $\mathbb{F}_q$ be the finite field of order $q$, where $q$ is a power of a prime $p$. In this paper, we always assume $\alpha_1, \ldots, \alpha_n$ are distinct elements of  $\mathbb{F}_q$ and $v_1, \ldots, v_n$ are nonzero elements of $\mathbb{F}_q$,
denoted by $\boldsymbol{\alpha}=(\alpha_1, \ldots, \alpha_n)$ and $\boldsymbol{v}=(v_1, \ldots, v_n)$.
In some specific cases, $\alpha_1, \ldots, \alpha_n$ will take distinct nonzero elements of $\mathbb{F}_q$.
For convenience, we denote $\boldsymbol{1}$ as the all-one vector,
$\boldsymbol{0}$ as the all-zero vector. The multiplication of two vectors ${\bf a}=(a_1,...,a_n), {\bf b}=(b_1,...,b_n)$ is defined as ${\bf a} \cdot {\bf b} = (a_1b_1,...,a_nb_n)$,
and their division is defined as ${{{\bf a}} \over {{\bf b}}} = \left(\frac{a_1}{b_1},...,\frac{a_n}{b_n}\right)$.

\begin{Definition}
For $0 \le n-t \le n$, the generalized Reed-Solomon (GRS) code is as follows:
$$
\GRS_{n-t}(\boldsymbol{\alpha}, \boldsymbol{v})=\left\{\left(v_1 f\left(\alpha_1\right), v_2 f\left(\alpha_2\right), \ldots, v_n f\left(\alpha_n\right)\right) \mid f(x) \in \mathbb{F}_q[x]_{n-t}\right\},
$$
where $\mathbb{F}_q[x]_{n-t}$ denotes the set of polynomials in $\mathbb{F}_q[x]$ of degree less than $n-t$,
which is a vector space of dimension $n-t$ over $\mathbb{F}_q$.
\end{Definition}

A GRS code $\GRS_{n-t}(\boldsymbol{\alpha}, \boldsymbol{v})$
is an $[n, n-t, t+1]$ linear code over $\mathbb{F}_q$, which has a generator matrix
$$
G=\small{\begin{pmatrix}
v_1 & \cdots & v_n \\
v_1 \alpha_1 & \cdots & v_n \alpha_n \\
\vdots & & \vdots \\
v_1 \alpha_1^{n-t-2} & \cdots & v_n \alpha_n^{n-t-2} \\
v_1 \alpha_1^{n-t-1} & \cdots & v_n \alpha_n^{n-t-1}
\end{pmatrix}}.
$$

In the references \cite{Beelen1, Guo, Huang2, Lavauzelle, Sui4, Wu3},
various forms of TGRS codes have been discussed. Below, we present the definitions of two types of TGRS codes.
\begin{Definition}\label{deftgrs1}
For $0 \le n-t \le n$, we define the twisted generalized Reed-Solomon (TGRS) code $C_1=\TGRS_{n-t,n-t}(\boldsymbol{\alpha}, \boldsymbol{v}, l, \eta_1, \lambda_1)$ over $\mathbb{F}_q$ with a generator matrix
\begin{align*}
    G_1 = \small{\begin{pmatrix}
        v_1 & v_2 & \cdots & v_n \\
        v_1 \alpha_1 & v_2 \alpha_2 & \cdots & v_n \alpha_n \\
        \vdots & \vdots & \ddots & \vdots \\
        v_1 \alpha_1^{l-1} & v_2 \alpha_2^{l-1} & \cdots & v_n \alpha_n^{l-1} \\
        v_1 \alpha_1^{l+1} & v_2 \alpha_2^{l+1} & \cdots & v_n \alpha_n^{l+1} \\
        \vdots & \vdots & \ddots & \vdots \\
        v_1 \alpha_1^{n-t-1} & v_2 \alpha_2^{n-t-1} & \cdots & v_n \alpha_n^{n-t-1} \\
        v_1(\lambda_1\alpha_1^{l}+\eta_1 \alpha_1^{n-t}) & v_2(\lambda_1\alpha_2^{l}+\eta_1 \alpha_2^{n-t}) & \cdots & v_n(\lambda_1\alpha_n^{l}+\eta_1 \alpha_n^{n-t})
    \end{pmatrix}},
\end{align*}
where $0 \le l \le n-t-1$, and either $\lambda_1\in \mathbb{F}_q$ or $\eta_1 \in \mathbb{F}_q$ is nonzero.
\end{Definition}

\begin{Definition}\label{deftgrs2}
For $0 \le n-t \le n$, we define the TGRS code $C_2=\TGRS_{n-t,-1}(\boldsymbol{\alpha}, \boldsymbol{v}, l, \eta_2, \lambda_2)$ over $\mathbb{F}_q$ with a generator matrix
\begin{align*}
    G_2 = \small{\begin{pmatrix}
    v_1 & v_2 & \cdots & v_n \\
    v_1 \alpha_1 & v_2 \alpha_2 & \cdots & v_n \alpha_n \\
    \vdots & \vdots & \ddots & \vdots \\
    v_1 \alpha_1^{l-1} & v_2 \alpha_2^{l-1} & \cdots & v_n \alpha_n^{l-1} \\
    v_1 \alpha_1^{l+1} & v_2 \alpha_2^{l+1} & \cdots & v_n \alpha_n^{l+1} \\
    \vdots & \vdots & \ddots & \vdots \\
    v_1 \alpha_1^{n-t-1} & v_2 \alpha_2^{n-t-1} & \cdots & v_n \alpha_n^{n-t-1} \\
    v_1(\lambda_2\alpha_1^{l}+\eta_2 \alpha_1^{-1}) & v_2(\lambda_2\alpha_2^{l}+\eta_2 \alpha_2^{-1}) & \cdots & v_n(\lambda_2\alpha_n^{l}+\eta_2 \alpha_n^{-1})
\end{pmatrix}},
\end{align*}
where $0 \le l \le n-t-1$, and either $\lambda_2\in \mathbb{F}_q$ or $\eta_2 \in \mathbb{F}_q$ is nonzero.
\end{Definition}
It is easy to see that $\TGRS_{n-t,n-t}( \boldsymbol{\alpha}, \boldsymbol{v}, l, \eta_1, \lambda_1)$ is a subcode of $\GRS_{n-t+1}( \boldsymbol{\alpha} ,\boldsymbol{v})$, and $\TGRS_{n-t,-1}( \boldsymbol{\alpha}, \boldsymbol{v}, l, \eta_2, \lambda_2)$ is a subcode of $\GRS_{n-t+1}( \boldsymbol{\alpha}, \boldsymbol{v} \cdot \boldsymbol{\alpha}^{-1})$.

\section{Parity-check matrices of TGRS codes}

For a code $C$ of length $n$ over $\mathbb{F}_q$, the dual code $C^{\bot}$ of $C$ is defined as $C^{\bot} = \{\bm{x} \in \mathbb{F}_q^n : \langle \bm{x},\bm{y} \rangle =
0$ for all $\bm{y} \in C\}$, where $\langle \bm{x},\bm{y} \rangle=\sum_{i=1}^nx_iy_i$ is the Euclidean (standard) inner product.

In this section, we determine the parity-check matrices of TGRS codes $C_1$ and $C_2$.
To obtain the general form of the parity-check matrices for $C_1$ and $C_2$,
we first present the well-known results for the parity-check matrix of a GRS code.

\begin{Proposition}
Assume the notation as given above. Then
$$\GRS_{t}(\boldsymbol\alpha,\boldsymbol{v})^\perp=\GRS_{n-t}(\boldsymbol{\alpha},\frac{\boldsymbol{u}}{\boldsymbol{v}}),$$
where $\boldsymbol{u} = (u_1, \ldots, u_n)$ with $u_i^{-1}= \prod_{j=1,j\ne i}^n(\alpha_i-\alpha_j), 1\le i \le n$.
\end{Proposition}

As we can see from the above,
$ \boldsymbol{u} \in \GRS_{n-2}(\boldsymbol\alpha,\boldsymbol{1})^\perp$.
Thus $\langle \boldsymbol{u}, \boldsymbol{\alpha}^i\rangle = 0$,
for $0 \le i \le n-2$ and $\langle \boldsymbol{u}, \boldsymbol{\alpha}^{n-1}\rangle \ne 0$.
If $\langle \boldsymbol{u}, \boldsymbol{\alpha}^{n-1}\rangle = 0$,
then it means that $\boldsymbol{u} \in \GRS_{n-1}(\boldsymbol\alpha,\boldsymbol{1})^\perp=\left(\mathbb{F}_q^n\right)^{\perp}$ and $\boldsymbol{u}=\boldsymbol{0}$.
This contradicts the definition of $\boldsymbol{u}$.
Similarly, when $\alpha_i$ is nonzero element of $\mathbb{F}_q (1\le i \le n)$, we have $\langle \boldsymbol{u}, \boldsymbol{\alpha}^{-1} \rangle \ne  0$.

In Definition \ref{deftgrs1}, when $\eta_1 = 0$, then $\TGRS_{n-t,n-t}(\boldsymbol{\alpha}, \boldsymbol{v}, l, 0, \lambda_1)$ is a GRS code.
Next, we consider the case $\eta_1 \ne 0$.

\begin{Theorem}\label{ch0}
    The code $\TGRS_{n-t,n-t}( \boldsymbol{\alpha}, \boldsymbol{v}, l, \eta_1, \lambda_1)(\eta_1 \ne 0, t > 1)$ has a parity-check matrix as follows:
$$
\small{\begin{pmatrix}
    \frac{u_1}{v_1} & \frac{u_2}{v_2} & \cdots & \frac{u_n}{v_n} \\
    \frac{u_1}{v_1} \alpha_1 & \frac{u_2}{v_2} \alpha_2 & \cdots & \frac{u_n}{v_n} \alpha_n \\
    \vdots & \vdots & \ddots & \vdots \\
    \frac{u_1}{v_1} \alpha_1^{t-2} & \frac{u_2}{v_2} \alpha_2^{t-2} & \cdots & \frac{u_n}{v_n} \alpha_n^{t-2} \\
    \frac{u_1}{v_1}(\alpha_1^{t-1} + f(\alpha_1)) & \frac{u_2}{v_2}(\alpha_2^{t-1} + f(\alpha_2)) & \cdots & \frac{u_n}{v_n}(\alpha_n^{t-1} + f(\alpha_n))
\end{pmatrix}},
$$
where
\begin{align}\label{ff0}
    f(x)=x^{n-l-1}+a_{n-l-2} x^{n-l-2}+\cdots+a_{t} x^{t}+a_{t-1}x^{t-1}\in \mathbb{F}_q[x]
\end{align}
with
\begin{align}\label{coe1}
    a_{n-l-1}=1,a_{n-l-2-r}=-\frac{\sum_{j=0}^{r} a_{n-l-1-j} \sum_{i=1}^n u_i \alpha_i^{n+r-j}}{\sum_{i=1}^n u_i \alpha_i^{n-1}}, \ \text{for} \ 0 \le r \le n-t-l-2,
\end{align}
and \[a_{t-1}=-\frac{\eta_1\sum_{j=0}^{n-t-l-1} a_{n-l-1-j} \sum_{i=1}^n u_i \alpha_i^{2n-t-l-1-j} + \lambda_1a_{n-l-1}\sum_{i=1}^nu_i\alpha_i^{n-1}}{\eta_1\sum_{i=1}^n u_i \alpha_i^{n-1}}-1.\]

\begin{proof}
	We know that $\langle \boldsymbol{u}, \boldsymbol{\alpha}^s  \rangle = 0$, for $0 \le s \le n-2$.
Thus $\langle \frac{\boldsymbol{u}}{\boldsymbol{v}}\boldsymbol{\alpha}^i, \boldsymbol{v}\boldsymbol{\alpha}^j \rangle =0$, for $0 \le i \le t-2$, and $0 \le j \le n-t$.
Therefore, $\frac{\boldsymbol{u}}{\boldsymbol{v}}\boldsymbol{\alpha}^i \in \TGRS_{n-t,n-t}( \boldsymbol{\alpha}, \boldsymbol{v}, l, \eta_1, \lambda_1)^{\bot}$, for $0 \le i \le t-2$.
We may consider non-zero polynomials of the form $f_1(x)=a_{t-1}x^{t-1}+\cdots+a_{n-1}x^{n-1}$, and then assume that $(\frac{u_1}{v_1}f_1(\alpha_1),..., \frac{u_n}{v_n}f_1(\alpha_n)) \in \TGRS_{n-t,n-t}( \boldsymbol{\alpha}, \boldsymbol{v}, l, \eta_1, \lambda_1)^{\bot}$.

The vector $\left(\frac{u_1}{v_1} f_1\left(\alpha_1\right), \cdots, \frac{u_n}{v_n} f_1\left(\alpha_n\right)\right)$ belongs to $\TGRS_{n-t,n-t}(\boldsymbol{\alpha}, \boldsymbol{v}, l, \eta_1, \lambda_1)^{\perp}$ if and only if the following system of equalities holds:
$$
\left\{\begin{array}{l}
\sum_{i=1}^n \frac{u_i}{v_i} f_1\left(\alpha_i\right) v_i=0, \\
\cdots \\
\sum_{i=1}^n \frac{u_i}{v_i} f_1\left(\alpha_i\right) v_i \alpha_i^{l-1}=0, \\
\sum_{i=1}^n \frac{u_i}{v_i} f_1\left(\alpha_i\right) v_i \alpha_i^{l+1}=0, \\
\cdots \\
\sum_{i=1}^n \frac{u_i}{v_i} f_1\left(\alpha_i\right) v_i \alpha_i^{n-t-1}=0, \\
\sum_{i=1}^n \frac{u_i}{v_i} f_1\left(\alpha_i\right) v_i\left(\lambda_1\alpha_i^{l}+\eta_1 \alpha_i^{n-t}\right)=0.
\end{array}\right.
$$
Since $\alpha_i \in \mathbb{F}_q^*(1 \le i \le n)$,
we can deduce that
$$
\left\{\begin{array}{l}
	a_{n-1}\sum_{i=1}^n u_i \alpha_i^{n-1} = 0, \\
	a_{n-2}\sum_{i=1}^{n-1} u_i \alpha_i^{n-1} + a_{n-1} \sum_{i=1}^{n-1} u_i \alpha_i^{n} = 0, \\
	\cdots\\
	a_{n-l}\sum_{i=1}^{n-1} u_i \alpha_i^{n-1} + a_{n-l+1}\sum_{i=1}^n u_i \alpha_i^{n} + \cdots + a_{n-1} \sum_{i=1}^n u_i \alpha_i^{n+l-2} = 0.
\end{array}\right.
$$
Then, we have
\[
	a_{n-1}=a_{n-2}= \cdots = a_{n-l}=0, f_1(x)=a_{t-1}x^{t-1}+\cdots+a_{n-l-1}x^{n-l-1},
\]
and
$$
\begin{cases}
    a_{n-l-1} \sum_{i=1}^n u_i \alpha_i^n+a_{n-l-2} \sum_{i=1}^n u_i \alpha_i^{n-1}=0, \\
a_{n-l-1} \sum_{i=1}^n u_i \alpha_i^{n+1}+a_{n-l-2} \sum_{i=1}^n u_i \alpha_i^n+a_{n-l-3} \sum_{i=1}^n u_i \alpha_i^{n-1}=0, \\
\cdots \\
a_{n-l-1} \sum_{i=1}^n u_i \alpha_i^{2n-t-l-2}+a_{n-l-2} \sum_{i=1}^n u_i \alpha_i^{2n-t-l-3}+\cdots+a_{t} \sum_{i=1}^n u_i \alpha_i^{n-1}=0, \\
\lambda_1a_{n-l-1} \sum_{i=1}^n u_i \alpha_i^{n-1} + \eta_1(a_{n-l-1} \sum_{i=1}^n u_i \alpha_i^{2n-t-l-1}+\cdots+a_{t-1} \sum_{i=1}^n u_i \alpha_i^{n-1})=0.
\end{cases}
$$

Note that $a_{n-l-1} \neq 0$. So we can assume $a_{n-l-1}=1$ by the linearity.
Since $\sum_{i=1}^n u_i \alpha_i^{n-1} \neq 0$, if $a_{n-l-1}=0$, then it follows from the first equality that $a_{n-l-2}=0$.
As a consequence of $a_{n-l-1}=a_{n-l-2}=0$, we have $a_{n-l-3}=0$ from the second equality.
Similarly, we can get $a_{n-l-4}=\cdots=a_{t-1}=0$ and hence $f_1(x)=0$,
which contradicts the assumption that $f_1(x)$ is non-zero.

So by solving the above system of equations, and by the assumption $a_{n-l-1}=1$,
we can obtain that the elements $a_i$ indeed satisfy the condition (\ref{coe1})
and
$$
    a_{t-1}=-\frac{\eta_1\sum_{j=0}^{n-t-l-1} a_{n-l-1-j} \sum_{i=1}^n u_i \alpha_i^{2n-t-l-1-j} + \lambda_1a_{n-l-1}\sum_{i=1}^nu_i\alpha_i^{n-1}}{\eta_1\sum_{i=1}^n u_i \alpha_i^{n-1}}.
$$
Let $f(x) = f_1(x) - x^{t-1}$. Then this completes the proof.
\end{proof}
\end{Theorem}

In Definition~\ref{deftgrs2}, when $\eta_2 = 0$, then $\TGRS_{n-t,-1}(\boldsymbol{\alpha}, \boldsymbol{v}, l, 0, \lambda)$ is a GRS code.
When $\lambda_2 = 0$, we may assume $\eta_2=1$ by the linearity.
In this case, it is easy to see $\TGRS_{n-t,-1}(\boldsymbol{\alpha}, \boldsymbol{v}, l, 1, 0)$ is equal to $\TGRS_{n-t,n-t}(\boldsymbol{\alpha}, \boldsymbol{v}\cdot\boldsymbol{\alpha}^{-1}, l, 1, 0)$.
Next, we consider the case where $\lambda_2 \ne 0$ and $\eta_2 \ne 0$.

\begin{Theorem}\label{ch1}
    The code $\TGRS_{n-t,-1}( \boldsymbol{\alpha}, \boldsymbol{v}, l, \eta_2, \lambda_2)(\eta_2 \ne 0, t > 1)$ has a parity-check matrix
$$
\small{\begin{pmatrix}
    \frac{u_1}{v_1} \alpha_1 & \frac{u_2}{v_2} \alpha_2 & \cdots & \frac{u_n}{v_n} \alpha_n \\
    \vdots & \vdots & \ddots & \vdots \\
    \frac{u_1}{v_1} \alpha_1^{t-1} & \frac{u_2}{v_2} \alpha_2^{t-1} & \cdots & \frac{u_n}{v_n} \alpha_n^{t-1} \\
    \frac{u_1}{v_1}(\alpha_1^{t} + f(\alpha_1)) & \frac{u_2}{v_2}(\alpha_2^{t} + f(\alpha_2)) & \cdots & \frac{u_n}{v_n}(\alpha_n^{t} + f(\alpha_n))
\end{pmatrix}},
$$
where
\begin{align}\label{ff1}
    f(x)=x^{n-l-1}+a_{n-l-2} x^{n-l-2}+\cdots+a_{t} x^{t}+a_0\in \mathbb{F}_q[x]
\end{align}
with
\begin{align}\label{coe2}
    a_{n-l-1}=1,a_{n-l-2-r}=-\frac{\sum_{j=0}^{r} a_{n-l-1-j} \sum_{i=1}^n u_i \alpha_i^{n+r-j}}{\sum_{i=1}^n u_i \alpha_i^{n-1}}, \ \text{for} \  0 \le r \le n-l-t-3,
\end{align}
$$a_t=-\frac{\sum_{j=0}^{n-l-t-2} a_{n-l-1-j} \sum_{i=1}^n u_i \alpha_i^{2n-l-t-2-j}}{\sum_{i=1}^n u_i \alpha_i^{n-1}}-1, \ \text{and} \ a_{0}=-\frac{\lambda_2  a_{n-l-1}\sum_{i=1}^nu_i\alpha_i^{n-1}}{\eta_2\sum_{i=1}^n u_i \alpha_i^{-1}}.$$

\end{Theorem}
\begin{proof}
	We know that $\langle \boldsymbol{u}, \boldsymbol{\alpha}^s\rangle = 0$, for $0 \le s \le n-2$.
Thus $\langle \frac{\boldsymbol{u}}{\boldsymbol{v}}\boldsymbol{\alpha}^i, \boldsymbol{v}\boldsymbol{\alpha}^j \rangle =0$, for $0 \le i \le t-2$ and $0 \le j \le n-t$.
Therefore, $\frac{\boldsymbol{u}}{\boldsymbol{v}}\boldsymbol{\alpha}^k \in \TGRS_{n-t, -1}( \boldsymbol{\alpha}, \boldsymbol{v}, l, \eta_2, \lambda_2)^{\bot}, 1\le k \le t-1$.
We may consider non-zero polynomials of the form $f_1(x)=a_0+a_{t}x^{t}+\cdots+a_{n-1}x^{n-1}$,
and then assume that $(\frac{u_1}{v_1}f_1(\alpha_1),..., \frac{u_n}{v_n}f_1(\alpha_n)) \in \TGRS_{n-t, -1}( \boldsymbol{\alpha}, \boldsymbol{v}, l, \eta_2, \lambda_2)^{\bot}$.

The vector $\left(\frac{u_1}{v_1} f_1\left(\alpha_1\right), \cdots, \frac{u_n}{v_n} f_1\left(\alpha_n\right)\right)$ belongs to $\TGRS_{n-t, -1}(\boldsymbol{\alpha}, \boldsymbol{v}, l, \eta_2, \lambda_2)^{\perp}$ if and only if the following system of equalities holds
$$
\left\{\begin{array}{l}
\sum_{i=1}^n \frac{u_i}{v_i} f_1\left(\alpha_i\right) v_i=0, \\
\cdots \\
\sum_{i=1}^n \frac{u_i}{v_i} f_1\left(\alpha_i\right) v_i \alpha_i^{l-1}=0, \\
\sum_{i=1}^n \frac{u_i}{v_i} f_1\left(\alpha_i\right) v_i \alpha_i^{l+1}=0, \\
\cdots \\
\sum_{i=1}^n \frac{u_i}{v_i} f_1\left(\alpha_i\right) v_i \alpha_i^{n-t-1}=0, \\
\sum_{i=1}^n \frac{u_i}{v_i} f_1\left(\alpha_i\right) v_i\left(\lambda_2\alpha_i^{l}+\eta_2 \alpha_i^{-1}\right)=0 .
\end{array}\right.
$$
Since $\alpha_i \in \mathbb{F}_q^*(1 \le i \le n)$, we can deduce that
$$
\begin{cases}
    a_{n-1}\sum_{i=1}^n u_i \alpha_i^{n-1} = 0, \\
	a_{n-2}\sum_{i=1}^{n-1} u_i \alpha_i^{n-1} + a_{n-1} \sum_{i=1}^{n-1} u_i \alpha_i^{n} = 0, \\
	\cdots\\
	a_{n-l}\sum_{i=1}^{n-1} u_i \alpha_i^{n-1} + a_{n-l+1}\sum_{i=1}^n u_i \alpha_i^{n} + \cdots + a_{n-1} \sum_{i=1}^n u_i \alpha_i^{n+l-2} = 0.
\end{cases}
$$
Then, we have
$$
	a_{n-1}=a_{n-2}= \cdots = a_{n-l}=0, f_1(x)=a_0 + a_{t}x^{t}+\cdots+a_{n-l-1}x^{n-l-1},
$$
and
$$
\begin{cases}
a_{n-l-1} \sum_{i=1}^n u_i \alpha_i^n+a_{n-l-2} \sum_{i=1}^n u_i \alpha_i^{n-1}=0, \\
a_{n-l-1} \sum_{i=1}^n u_i \alpha_i^{n+1}+a_{n-l-2} \sum_{i=1}^n u_i \alpha_i^n+a_{n-l-3} \sum_{i=1}^n u_i \alpha_i^{n-1}=0, \\
\cdots \\
a_{n-l-1} \sum_{i=1}^n u_i \alpha_i^{2n-t-l-2}+a_{n-l-2} \sum_{i=1}^n u_i \alpha_i^{2n-t-l-3}+\cdots+a_{t} \sum_{i=1}^n u_i \alpha_i^{n-1}=0, \\
\lambda_2a_{n-l-1} \sum_{i=1}^n u_i \alpha_i^{n-1} + \eta_2a_0\sum_{i=1}^nu_i\alpha_i^{-1}=0.
\end{cases}
$$

Note that $a_{n-l-1} \neq 0$. So we can assume $a_{n-l-1}=1$ by the linearity.
Since $\sum_{i=1}^n u_i \alpha_i^{n-1} \neq 0$, if $a_{n-l-1}=0$, then it follows from the first equality that $a_{n-l-2}=0$.
As a consequence of $a_{n-l-1}=a_{n-l-2}=0$, we have $a_{n-l-3}=0$ from the second equality.
Similarly, we can get $a_{n-l-4}=\cdots=a_{t}=a_0=0$ and hence $f_1(x)=0$, which contradicts the assumption that $f_1(x)$ is non-zero.

So by solving the above system of equations, and by assumption $a_{n-l-1}=1$,
we can obtain that the elements $a_i$ indeed satisfy the condition (\ref{coe2}) and
\[
    a_t=-\frac{\sum_{j=0}^{n-l-t-2} a_{n-l-1-j} \sum_{i=1}^n u_i \alpha_i^{2n-l-t-2-j}}{\sum_{i=1}^n u_i \alpha_i^{n-1}}, a_{0}=-\frac{\lambda_2  a_{n-l-1}\sum_{i=1}^nu_i\alpha_i^{n-1}}{\eta_2\sum_{i=1}^n u_i \alpha_i^{-1}}.
\]
Let $f(x) = f_1(x) - x^t$. This completes the proof.
\end{proof}

We provide here the general forms of the parity-check matrices for codes $C_1$ and $C_2$,
and we will utilize these matrices in subsequent steps for decoding.

\begin{Remark}
    Based on the results discussed above,
it can be concluded that $\TGRS_{n-t,n-t}( \boldsymbol{\alpha}, \boldsymbol{v}, l, \eta_1, \lambda_1)$ and $\TGRS_{n-t,-1}( \boldsymbol{\alpha}, \boldsymbol{v}, l, \eta_2, \lambda_2)$ are either MDS codes or almost-MDS codes.
\end{Remark}

\section{Decoding}
In \cite{Sun1}, Sun et al. discussed the decoding issues associated with two specific MDS TGRS codes.
However, a notable limitation is the overly stringent conditions that must be met for TGRS codes to be classified as MDS codes ({see \cite[Lemma 2.2]{Sun1}}).
To address this limitation, we have embarked on research aimed at decoding a more general range of TGRS codes,
adopting distinct processing strategies for MDS and almost-MDS TGRS codes respectively.

From now on, we always assume that $\alpha_1,...,\alpha_n$ are all distinct nonzero elements of $\mathbb{F}_q$.
In this section, we consider the decoding of $\TGRS_{n-t,n-t}(\boldsymbol{\alpha}, \boldsymbol{v}, l, \eta_1, \lambda_1)$ and
$\TGRS_{n-t,-1}(\boldsymbol{\alpha}, \boldsymbol{v}, l, \eta_2, \lambda_2)$.
Firstly, we focus on the decoding of a more general class of TGRS codes.

Let $C$ be an $[n, n-t, d]$, ($d = t$ or $d=t+1$) TGRS code with parity-check matrix as
\begin{align}\label{parh}
    H =\begin{pmatrix}
        w_1 & w_2 & \cdots & w_n \\
        w_1 \alpha_1 & w_2 \alpha_2 & \cdots & w_n \alpha_n \\
        \vdots & \vdots & \ddots & \vdots \\
        w_1 \alpha_1^{t-2} & w_2 \alpha_2^{t-2} & \cdots & w_n \alpha_n^{t-2} \\
        w_1(\alpha_1^{t-1} + f(\alpha_1)) & w_2(\alpha_2^{t-1} + f(\alpha_2)) & \cdots & w_n(\alpha_n^{t-1} + f(\alpha_n))
    \end{pmatrix},
\end{align}
where $f(x) \in \mathbb{F}_q[x]$, and $w_1, \ldots, w_n$ are nonzero elements of $\mathbb{F}_q$.

\begin{Remark}\label{equrem}
    According to Theorem~\ref{ch0}, when $\boldsymbol{w}=(w_1,...,w_n)$ is set to $(\frac{u_1}{v_1},...,\frac{u_n}{v_n})$ and $f(x)$ is as given in (\ref{ff0}),
then the code $C$ is equal to $\TGRS_{n-t,n-t}(\boldsymbol{\alpha}, \boldsymbol{v}, l, \eta_1, \lambda_1)$ code.
From Theorem~\ref{ch1}, when $\boldsymbol{w}=(w_1,...,w_n)$ is set to $(\frac{u_1}{v_1}\cdot \alpha_1,...,\frac{u_n}{v_n}\cdot\alpha_n)$ and $f(x)$ is taken as $x^{q-2}\cdot f(x)$ as given in (\ref{ff1}),
then the code $C$ is equal to $\TGRS_{n-t,-1}(\boldsymbol{\alpha}, \boldsymbol{v}, l, \eta_2, \lambda_2)$ code.
\end{Remark}

It is evident that the code $C$ is either an MDS code or an almost-MDS code.
We shall give the key equations of $C$ for decoding.

Let $\boldsymbol{r}=\left(r_1, \cdots, r_n\right)$ be a received word with
    $\boldsymbol{r}=\boldsymbol{c}+\boldsymbol{e}$,
where $\boldsymbol{c}=\left(c_1, \cdots, c_n\right)$ is a codeword of $C$,
$\boldsymbol{e}=\left(e_1, \cdots, e_n\right)$ is an error word,
and $J=\left\{j\,|\, 1 \le j \le n, e_j \neq 0\right\}$ is called the error location set with $|J| \le \lfloor \frac{d-1}{2} \rfloor$,
where  $|J|$ denotes the number of elements in the set $J$.

\textbf{Case 1:} $d=t$, i.e., $C$ is an almost-MDS code, or $d=t+1$ and $t$ is odd.

In these two situations, we can use a submatrix $H_1$ of the parity-check matrix $H$ of $C$ for decoding,
where
\begin{align}\label{submala}
    H_1=\begin{pmatrix}
        w_1 & w_2 & \cdots & w_n \\
        w_1 \alpha_1 & w_2 \alpha_2 & \cdots & w_n \alpha_n \\
        \vdots & \vdots & \ddots & \vdots \\
        w_1 \alpha_1^{t-2} & w_2 \alpha_2^{t-2} & \cdots & w_n \alpha_n^{t-2}
    \end{pmatrix}.
\end{align}
Let the syndrome of $\boldsymbol{r}$ be
$$
\boldsymbol{s}=\begin{pmatrix}
s_0 \\
s_1 \\
\vdots \\
s_{t-2}
\end{pmatrix}
=H_1 \boldsymbol{r}^T=H_1 \boldsymbol{c}^T+H_1 \boldsymbol{e}^T=H_1 \boldsymbol{e}^T,
$$
where
$$
 s_i=\sum_{j \in J} e_jw_j\alpha_j^i, 0 \le i \le t-2.
$$
Define the syndrome polynomial $S(x)$ of the received word $\boldsymbol{r}$:
\begin{align}
    \begin{aligned} \label{ss0}
        S(x) & =\sum_{i=0}^{t-2} s_i x^i=\sum_{i=0}^{t-2} \sum_{j \in J} e_jw_j  \alpha_j^i x^i \\
        & = \sum_{j \in J}\sum_{i=0}^{t-2} e_j w_j \alpha_j^{i} x^i \\
        & =\sum_{j \in J} e_jw_j  \frac{1-(\alpha_jx)^{t-1}}{1-(\alpha_jx)} \\
        & \equiv-\sum_{j \in J} e_jw_j\frac{\alpha_j^{-1}}{x-\alpha_j^{-1}} \pmod{x^{t-1}}.
    \end{aligned}
\end{align}
The error location polynomial is
$$
\sigma(x)=\prod_{j \in J}\left(x-\alpha_j^{-1}\right)
$$
and the error evaluator polynomial is
$$
\begin{aligned}
\tau(x) & =\left(-\sum_{i \in J} e_iw_i\frac{\alpha_i^{-1}}{x-\alpha_i^{-1}}\right) \prod_{j \in J}\left(x-\alpha_j^{-1}\right) \\
& =-\sum_{i \in J} e_iw_i \alpha_i^{-1} \prod_{j \in J \backslash\{i\}}\left(x-\alpha_j^{-1}\right) .
\end{aligned}
$$
Then
\begin{align}\label{eq0}
    S(x) \sigma(x) \equiv \tau(x) \pmod{x^{t-1}}.
\end{align}
It is clear that
\begin{align}\label{eq1}
    \gcd(\sigma(x), \tau(x))=1, \deg \tau(x) < \deg \sigma(x)=|J| \le \lfloor \frac{d-1}{2} \rfloor.
\end{align}
For each $i \in J$,
\[
    \tau(\alpha_i^{-1})=-e_iw_i\alpha_i^{-1} \prod_{j \in J\backslash\{i\}}\left(\alpha_i^{-1}-\alpha_j^{-1}\right)=-e_iw_i\alpha_i^{-1}\sigma'(\alpha_i^{-1}), e_i=-\frac{\alpha_i\tau(\alpha_i^{-1})}{w_i\sigma'(\alpha_i^{-1})},
\]
where $\sigma'(x)$ is the formal derivative of $\sigma(x)$.
\begin{Theorem}\label{unithm1}
    Let $C$ be a TGRS $[n, n-t, d]$ code with $d=t$ or ($d=t+1$ and $t$ is odd).
Let $\boldsymbol{r}$ be a received word with $d(\boldsymbol{r}, C) \le \lfloor \frac{d-1}{2} \rfloor$ and $S(x)$ the syndrome polynomial of $\boldsymbol{r}$ as (\ref{ss0}).
Then there is a unique polynomial pair $(\sigma(x), \tau(x))$ in Equations (\ref{eq0})-(\ref{eq1}) up to the leading coefficient of $\sigma(x)$.

\begin{proof}
    Assume there exist two pairs $(\sigma^{(1)}(x), \tau^{(1)}(x))$ and $(\sigma^{(2)}(x), \tau^{(2)}(x))$ that satisfy
Equations (\ref{eq0})-(\ref{eq1}). i.e.,
\begin{align*}
    S(x) \sigma^{(1)}(x) \equiv \tau^{(1)}(x) \pmod{x^{t-1}}, S(x) \sigma^{(2)}(x) \equiv \tau^{(2)}(x) \pmod{x^{t-1}}.
\end{align*}
Given that $\sigma^{(1)}(x)\ne 0$ and $\sigma^{(2)}(x)\ne 0$, then
\[
    \sigma^{(2)}(x)\tau^{(1)}(x)  \equiv \sigma^{(1)}(x)\tau^{(2)}(x) \pmod{x^{t-1}}.
\]
Since $\deg(\tau^{(1)}(x)) < \deg(\sigma^{(1)}(x)) \le \lfloor \frac{d-1}{2} \rfloor = \lfloor \frac{t-1}{2} \rfloor$ and $\deg(\tau^{(2)}(x)) < \deg(\sigma^{(2)}(x)) \le \lfloor \frac{d-1}{2} \rfloor= \lfloor \frac{t-1}{2} \rfloor$,
we have
\[
    \sigma^{(2)}(x)\tau^{(1)}(x)  = \sigma^{(1)}(x)\tau^{(2)}(x).
\]
Moreover, note that $\gcd(\sigma^{(1)}(x), \tau^{(1)}(x))=1$ and $\gcd(\sigma^{(2)}(x), \tau^{(2)}(x))=1$, we conclude that
\[
    \sigma^{(1)}(x)=\lambda \sigma^{(2)}(x), \tau^{(1)}(x) =\lambda \tau^{(2)}(x), \lambda \in \mathbb{F}_q^*.
\]
Therefore, up to the leading coefficient of $\sigma(x)$, there is a unique pair $(\sigma(x), \tau(x))$.
\end{proof}
\end{Theorem}

Based on the above theorem, Equations (\ref{eq0}) and (\ref{eq1}) form the key equations of TGRS code $C$ for Case 1.

\textbf{Case 2}:
$d=t+1$ and $t$ is even.
Let the syndrome of $\boldsymbol{r}$ be
$$
\boldsymbol{s}=\begin{pmatrix}
    s_0 \\
s_1 \\
\vdots \\
s_{t-1}
\end{pmatrix}=H \boldsymbol{r}^T=H \boldsymbol{c}^T+H \boldsymbol{e}^T=H \boldsymbol{e}^T,
$$
where
$$
 s_i=\sum_{j \in J} e_jw_j  \alpha_j^i (0 \le i \le t-2), s_{t-1}=\sum_{j \in J} e_jw_j  \left(\alpha_j^{t-1} +  f(\alpha_j)\right) .
$$
Define the syndrome polynomial $S(x)$ of the received word $\boldsymbol{r}$ as
\begin{align}
\begin{aligned}\label{ss1}
S(x) & =\sum_{i=0}^{t-1} s_{t-i-1} x^i=\sum_{i=1}^{t-1} \sum_{j \in J} e_jw_j \alpha_j^{t-i-1} x^i+\sum_{j \in J} e_jw_j \left(\alpha_j^{t-1}+f(\alpha_j)\right) \\
& =\sum_{i=0}^{t-1} \sum_{j \in J} e_jw_j \alpha_j^{t-i-1} x^i+\sum_{j \in J} e_jw_j f(\alpha_j) \\
& =\sum_{j \in J} e_jw_j  \frac{x^t-\alpha_j^t}{x-\alpha_j}+\sum_{j \in J} e_jw_j f(\alpha_j) \\
& \equiv-\sum_{j \in J} w_j\frac{ e_j \alpha_j^t}{x-\alpha_j}+\sum_{j \in J} e_jw_j f(\alpha_j) \pmod{x^t}.
\end{aligned}
\end{align}
The error location polynomial is
$$
    \sigma(x)=\prod_{j \in J}\left(x-\alpha_j\right)
$$
and the error evaluator polynomial is
$$
\begin{aligned}
\tau(x) & =\left(-\sum_{i \in J} \frac{ e_iw_i \alpha_i^t}{x-\alpha_i}+\sum_{i \in J} e_i w_i f(\alpha_i)\right) \prod_{j \in J}\left(x-\alpha_j\right) \\
& =\sigma(x) \sum_{i \in J} e_iw_i f(\alpha_i)-\sum_{i \in J} e_iw_i \alpha_i^t \prod_{j \in J \backslash\{i\}}\left(x-\alpha_j\right) .
\end{aligned}
$$
Then
\begin{align}\label{sxx1}
    S(x) \sigma(x) \equiv \tau(x) \pmod{x^t}.
\end{align}
It is clear that
\begin{align}\label{sxx11}
    \gcd(\sigma(x), \tau(x))=1, \deg \tau(x) \le \deg \sigma(x)=|J| \le \frac{t}{2}.
\end{align}
By division with remainder,
\[
    \tau(x)=a \sigma(x)+\omega(x), a=\sum_{j \in J} e_jw_j f(\alpha_j) \in \mathbb{F}_q,\omega(x)=-\sum_{j\in J}e_jw_j\alpha_j^t\frac{\sigma(x)}{x-\alpha_j},
\]
where $\deg \omega(x)<\deg \sigma(x)$. For each $i \in J$,
\[
    \tau\left(\alpha_i\right)=-e_iw_i \alpha_i^t \prod_{j \in J \backslash\{i\}}\left(\alpha_i-\alpha_j\right)=-e_iw_i \alpha_i^t \sigma^{\prime}\left(\alpha_i\right), e_i=-\frac{ \tau\left(\alpha_i\right)}{w_i\alpha_i^t \sigma^{\prime}\left(\alpha_i\right)}.
\]
Here, $\sigma^{\prime}(x)$ is the formal derivative of $\sigma(x)$.

The relationship between $\tau(x)$ and $\sigma(x)$ is as follows:
\begin{align}\label{sxx2}
    \begin{aligned}
        J &=\{i\,|\,\sigma(\alpha_i)=0,1 \le i \le n\}, \deg \sigma(x)=|J|, \\
        e_i& = \begin{cases}
            -\frac{ \tau\left(\alpha_i\right)}{w_i\alpha_i^t \sigma^{\prime}\left(\alpha_i\right)},  \ \text{if} \ i \in J, \\
            0, \ \text{if} \ i \notin J.
        \end{cases}\\
        \tau(x)&=a\sigma(x)+\omega(x), a=\sum_{j \in J}  e_jw_j f(\alpha_j),\omega(x)=-\sum_{j\in J}e_jw_j\alpha_j^t\frac{\sigma(x)}{x-\alpha_j}, \deg \omega(x) < |J|.
    \end{aligned}
\end{align}

\begin{Theorem}\label{unithm3}
   Let $C$ be an MDS TGRS $[n, n-t, t+1]$ code with $t$ even.
Let $\boldsymbol{r}$ be a received word with $d\left(\boldsymbol{r}, C\right) \le \frac{t}{2}$, and let $S(x)$ be the syndrome polynomial of $\boldsymbol{r}$ as in (\ref{ss1}).
Then there is a unique polynomial pair $(\sigma(x), \tau(x))$ satisfying Equations (\ref{sxx1})-(\ref{sxx2}), up to the leading coefficient of $\sigma(x)$.
\begin{proof}
    We prove this theorem by two subcases.

    \textbf{Subcase 1}: $d(\boldsymbol{r}, C) < \frac{t}{2}$.
    Assume there exist two pairs $(\sigma^{(1)}(x), \tau^{(1)}(x))$ and $(\sigma^{(2)}(x), \tau^{(2)}(x))$ that satisfy
    Equations (\ref{sxx1})-(\ref{sxx2}). i.e.,
    \begin{align*}
        S(x) \sigma^{(1)}(x) \equiv \tau^{(1)}(x) \pmod{x^{t}}, S(x) \sigma^{(2)}(x) \equiv \tau^{(2)}(x) \pmod{x^{t}}.
    \end{align*}
    It is clear that $\sigma^{(1)}(x)\ne 0$ and $\sigma^{(2)}(x)\ne 0$. Hence
    \[
        \sigma^{(2)}(x)\tau^{(1)}(x)  \equiv \sigma^{(1)}(x)\tau^{(2)}(x) \pmod{x^{t}}.
    \]
    Since $\deg(\tau^{(1)}(x)) \le \deg(\sigma^{(1)}(x)) < \frac{t}{2}$ and $\deg(\tau^{(2)}(x)) \le \deg(\sigma^{(2)}(x)) < \frac{t}{2}$,
    we have
    \[
        \sigma^{(2)}(x)\tau^{(1)}(x)  = \sigma^{(1)}(x)\tau^{(2)}(x).
    \]
    Furthermore, since $\gcd(\sigma^{(1)}(x), \tau^{(1)}(x))=1$ and $\gcd(\sigma^{(2)}(x), \tau^{(2)}(x))=1$, we conclude that
    \[
        \sigma^{(1)}(x)=\lambda \sigma^{(2)}(x), \tau^{(1)}(x) =\lambda \tau^{(2)}(x), \lambda \in \mathbb{F}_q^*.
    \]
    Therefore, up to the leading coefficient of $\sigma(x)$, there is a unique pair $(\sigma(x), \tau(x))$.

    \textbf{Subcase 2}: $d(\boldsymbol{r}, C) = \frac{t}{2}$.
    Assume there exist two pairs $(\sigma^{(1)}(x), \tau^{(1)}(x))$ and $(\sigma^{(2)}(x), \tau^{(2)}(x))$ that satisfy
    Equations (\ref{sxx1})-(\ref{sxx2}). Without loss of generality, let $\sigma^{(1)}(x) =\prod_{j \in J_1}\left(x-\alpha_j\right)$ and $\sigma^{(2)}(x) =\prod_{j \in J_2}\left(x-\alpha_j\right)$.
    Then
    \begin{align*}
        \tau^{(1)}(x)&=a_1\sigma^{(1)}(x)+\omega_1(x), a_1=\sum_{j \in J_1}  e_jw_j f(\alpha_j),\omega_1(x)=\sum_{j\in J_1}e_jw_j\alpha_j^t\frac{\sigma^{(1)}(x)}{x-\alpha_j}, \deg \omega_1(x) < |J_1|, \\
        \tau^{(2)}(x)&=a_2\sigma^{(2)}(x)+\omega_2(x), a_2=\sum_{j \in J_2}  e'_jw_j f(\alpha_j),\omega_2(x)=\sum_{j\in J_2}e'_jw_j\alpha_j^t\frac{\sigma^{(2)}(x)}{x-\alpha_j}, \deg \omega_2(x) < |J_2|.
    \end{align*}
    Thus,
    \begin{align*}
        \tau^{(1)}(x) & =\left(-\sum_{j \in J_1} \frac{ e_jw_j \alpha_j^t}{x-\alpha_j}+\sum_{j \in J_1} e_j w_j f(\alpha_j)\right) \sigma^{(1)}(x), \\
        \tau^{(2)}(x) & =\left(-\sum_{j \in J_2} \frac{ e'_jw_j \alpha_j^t}{x-\alpha_j}+\sum_{j \in J_2} e'_j w_j f(\alpha_j)\right) \sigma^{(2)}(x).
    \end{align*}
    By the conditions
    \[
        S(x) \sigma^{(1)}(x) \equiv \tau^{(1)}(x) \pmod{x^t},
        S(x) \sigma^{(2)}(x) \equiv \tau^{(2)}(x) \pmod{x^t},
    \]
    we have
    \[
        S(x) \equiv \left(-\sum_{j \in J_1} \frac{ e_jw_j \alpha_j^t}{x-\alpha_j}+\sum_{j \in J_1} e_j w_j f(\alpha_j)\right) \equiv \left(-\sum_{j \in J_2} \frac{ e'_jw_j \alpha_j^t}{x-\alpha_j}+\sum_{j \in J_2} e'_j w_j f(\alpha_j)\right) \pmod{x^t}.
    \]
    Then
    \begin{align} \label{proofcor}
        \begin{aligned}
            S(x) &\equiv \sum_{i=0}^{t-1} \sum_{j \in J_1} e_jw_j \alpha_j^{t-i-1} x^i+\sum_{j \in J_1} e_jw_j f(\alpha_j) \pmod{x^t} \\
            &\equiv \sum_{i=0}^{t-1} \sum_{j \in J_2} e'_jw_j \alpha_j^{t-i-1} x^i+\sum_{j \in J_2} e'_jw_j f(\alpha_j) \pmod{x^t}.
        \end{aligned}
    \end{align}

   Since $C$ is an $[n, n-t, t+1]$ MDS code and $|J_1|=|J_2|=\frac{t}{2}$,
    Equation (\ref{proofcor}) has a unique solution.
    Thus $J_1 = J_2$ and $e_j = e_j'$ for any $j \in J_1$.
    Up to the leading coefficient of $\sigma(x)$, there is a unique pair $(\sigma(x), \tau(x))$.
\end{proof}
\end{Theorem}

Equations (\ref{sxx1})-(\ref{sxx2}) form the key equations of TGRS code $C$ for Case 2.

By Remark~\ref{equrem}, the results in this section are applicable to the codes $\TGRS_{n-t,n-t}(\boldsymbol{\alpha}, \boldsymbol{v}, l, \eta_1, \lambda_1)$ and $\TGRS_{n-t,-1}(\boldsymbol{\alpha}, \boldsymbol{v}, l, \eta_2, \lambda_2)$.
To eliminate confusion, we will only discuss the decoding of $\TGRS_{n-t,n-t}(\boldsymbol{\alpha}, \boldsymbol{v}, l, \eta_1, \lambda_1)$.

\section{Decoding for TGRS codes}

The Berlekamp-Massey Algorithm \cite{Berlekamp2} has achieved many successful applications in engineering.
In \cite{Sugiyama}, Sugiyama was the first researcher to successfully utilize the Euclid's Algorithm for decoding GRS and Goppa codes,
\cite{Sui3} and \cite{Sun1} also considered the decoding of TGRS codes using similar methods.

In Section 4, we have explored the decoding problem of a class of TGRS codes and attributed the uniqueness of decoding to the uniqueness of the error location polynomial $\sigma(x)$ and the error evaluator polynomial $\tau(x)$ under certain conditions.

In this section, we shall use the extended Euclid's Algorithm to construct all possible polynomial pairs $(\sigma(x), \tau(x))$ to ensure that they satisfy the conditions stated in Theorems \ref{unithm1} and \ref{unithm3}, respectively.

\subsection{Extended Euclid's Algorithms}
The extended Euclid's Algorithm, tailored for polynomials over the finite field $\mathbb{F}_q$, serves as a potent method for solving key equations by facilitating the computation of the greatest common divisor (GCD) of two polynomials, $g(x)$ and $S(x)$, with $g(x) \neq 0$ and $\deg g(x) > \deg S(x)$.

This algorithm iteratively computes:
remainders, denoted by $\tau_i(x)$, quotients, denoted by $q_i(x)$,
auxiliary polynomial, $\sigma_i(x)$.
The initial setup for these polynomials is established as:
\[
\begin{aligned}
& \sigma_{-1}(x) = 0, \quad \tau_{-1}(x) = g(x), \\
& \sigma_0(x) = 1, \quad \tau_0(x) = S(x).
\end{aligned}
\]
Subsequently, for each step $i$, the quotient $q_i(x)$ and the next remainder $\tau_i(x)$ are determined by the division of $\tau_{i-2}(x)$ by $\tau_{i-1}(x)$:
\[
\tau_{i-2}(x) = q_i(x) \tau_{i-1}(x) + \tau_i(x), \quad \text{where} \quad \deg \tau_i(x) < \deg \tau_{i-1}(x).
\]
Concurrently, the auxiliary polynomial $\sigma_i(x)$ is updated using the following relations:
\[
\sigma_i(x) = \sigma_{i-2}(x) - q_i(x) \sigma_{i-1}(x).
\]

Let $v$ represent the largest index for which $\tau_v(x) \neq 0$. It is a well-established fact that:
\[
    \tau_v(x) = \gcd(S(x), g(x)).
\]

In other words, the non-zero remainder with the smallest degree, obtained through the iterative process of the extended Euclid's Algorithm, is the greatest common divisor of the polynomials $S(x)$ and $g(x)$.

The following theorem represents the main result required by the Sugiyama Algorithm \cite{Sugiyama}.
Additionally, the conclusion presented can be directly utilized in the context of Case 1 of TGRS code decoding, as discussed in Section 4.

\begin{Theorem}\cite{Sugiyama}
Let $\sigma _{i}( x)$ and $\tau_i(x)$ for $i\in \{ - 1, 0, . . . , v + 1\}$ be polynomials from the Euclid's Algorithm applied to $g(x)$ and $S(x)$.
Suppose that $\sigma (x)$ and $\tau (x)$ are nonzero polynomials over $\mathbb{F}_{q}$ satisfying the following conditions:

\noindent (1) $\gcd (\sigma(x), \tau(x)) = 1$,

\noindent (2) $\deg\sigma(x)+\deg\tau(x)<\deg g(x)$,

\noindent (3) $\sigma (x) S(x) \equiv \tau (x) \pmod{g(x)}$.

\noindent Then there is a  unique index $h\in \{ 0, 1, . . . , v + 1\}$ and a constant $\lambda \in \mathbb{F} _q$ such that
$$\sigma(x)=\lambda\sigma_h(x),\tau(x)=\lambda\tau_h(x).$$
Moreover, if $\deg\sigma(x)\le \frac{1}{2}\deg g(x)$,
and $\deg\tau(x) < \frac{1}{2} \deg g(x)$,
then the value $h$ is the unique index for which the remainders in the Euclid's Algorithm satisfy
$\deg \tau_h < \frac{1}{2} \deg g \le \deg \tau_{h-1}$.
\end{Theorem}

\begin{Theorem} \label{abc} \cite{Sun1}
    Let $g(x)$ and $S(x)$ be two polynomials with $\deg S(x)<\deg g(x)=t$, where $t$ is even.
Let $\sigma_i(x)$ and $\tau_i(x)$ for $i \in\{-1,0, \ldots, v+1\}$ be the polynomials from the Euclid's Algorithm applied to $g(x)$ and $S(x)$.
Suppose that there is a polynomial pair $(\sigma(x), \tau(x))$ over $\mathbb{F}_q$ that satisfies the following conditions:

\noindent (1) $\gcd(\sigma(x), \tau(x))=1$,

\noindent (2) $\deg \tau(x) \le \deg \sigma(x)=\frac{t}{2}$,

\noindent (3) $\sigma(x) S(x) \equiv \tau(x) \pmod{g(x)}$.

\noindent Then there are $\lambda_1 \in \mathbb{F}_q$ and $\lambda_2 \in \mathbb{F}_q^*$ such that
$$
    \sigma(x)=\lambda_1 \sigma_{h-1}(x)+\lambda_2 \sigma_h(x), \tau(x)=\lambda_1 \tau_{h-1}(x)+\lambda_2 \tau_h(x),
$$
where $\tau_h(x)$ is the polynomial which has the minimum index $h \in\{0,1, \cdots, v+1\}$ and satisfies $\deg \tau_h(x)<\frac{t}{2}$. Moreover, if $\deg \tau(x)<\deg \sigma(x)=\frac{t}{2}$ in (2), then $\lambda_1=0$.
\end{Theorem}

Theorem~\ref{unithm3} implies that a unique polynomial pair $(\sigma(x), \tau(x))$ satisfying Equations (\ref{sxx1})-(\ref{sxx2}) exists.
Theorem~\ref{abc} provides the specific form of this polynomial pair $(\sigma(x), \tau(x))$ that satisfies Equations (\ref{sxx1})-(\ref{sxx11}).
In the following, we present a more detailed result regarding the polynomial pair $(\sigma(x), \tau(x))$,
which will help us optimize the performance of the decoding algorithm for TGRS codes.

\begin{Theorem}\label{newthm}
    Under the conditions of Theorem \ref{abc}. Then there are $\lambda_1 \in \mathbb{F}_q$ and $\lambda_2 \in \mathbb{F}_q^*$ such that
$$
    \sigma(x)=\lambda_2 (\lambda_1 \sigma_{h-1}(x)+\sigma_h(x)), \tau(x)=\lambda_2 (\lambda_1 \tau_{h-1}(x)+ \tau_h(x)),
$$
where $\tau_h(x)$ is the polynomial which has the minimum index $h \in\{0,1, \cdots, v+1\}$ and satisfies $\deg \tau_h(x)<\frac{t}{2}$.
Moreover, $\lambda_1$ is one of the most frequent elements in the set $\mathcal{B}$,
where
\begin{align}\label{betadef}
    \mathcal{B}=\{\beta_i | i \le i \le n \}\backslash \infty \ \text{and} \ \beta_i = \begin{cases}
        \sigma_{h-1}(\alpha_i)^{-1} \sigma_{h}(\alpha_i), \ & \text{if} \ \sigma_{h-1}(\alpha_i) \ne 0, \\
        \infty,  & \text{if} \ \sigma_{h-1}(\alpha_i) = 0,
    \end{cases}
\end{align}
for $1 \le i \le n$.

\begin{proof}
By Theorems~\ref{abc} and \ref{unithm3},
if polynomial pair $(\sigma(x), \tau(x))$ satisfies Conditions (1)-(3),
there is unique  $\lambda_1 \in \mathbb{F}_q$ and $\lambda_2 \in \mathbb{F}_q^*$ such that
$$\sigma(x)=\lambda_2 (\lambda_1 \sigma_{h-1}(x)+\sigma_h(x)), \tau(x)=\lambda_2 (\lambda_1 \tau_{h-1}(x)+ \tau_h(x)).$$
We define
\[
    \sigma_{\lambda}(x)=\lambda \sigma_{h-1}(x)+\sigma_h(x).
\]
For fixed $i \in \{1,...,n\}$:

If $\sigma_{h-1}(\alpha_i) = 0$ and $\sigma_{h}(\alpha_i) = 0$, then for any $\lambda \in \mathbb{F}_q$, $\sigma_{\lambda}(\alpha_i) = 0$.

If $\sigma_{h-1}(\alpha_i) = 0$ and $\sigma_{h}(\alpha_i) \ne 0$, then for any $\lambda \in \mathbb{F}_q$, $\sigma_{\lambda}(\alpha_i) \ne 0$.

\noindent Let
\[
    N_0 =|\{i| 1 \le i \le n, \sigma_{n-1}(\alpha_i)=0, \sigma_n(\alpha_i)=0\}|,
\]
and let
\[
    N(\beta) = |\{i | \beta_i = \beta, 1 \le i \le n\}|,\ \text{where $\beta_i$ is defined as (\ref{betadef})}.
\]

Then, the polynomial $\sigma_{\lambda}(x)$ has $N_0 + N(\lambda)$ roots (without counting multiplicities) in the set $\{\alpha_i | 1 \le i \le n\}$.
When $\lambda$ takes the value of a most frequently occurring element in $\mathcal{B}$,
the polynomial $\sigma_{\lambda}(x)$ has the largest number of roots in the set $\{\alpha_i | 1 \le i \le n\}$.
Since $\sigma_{\lambda_1}(x)$ has $\deg(\sigma_{\lambda_1}(x))=\frac{t}{2}$ roots in $\{\alpha_i | 1 \le i \le n\}$, $\deg\sigma_{\lambda}(x) \le \frac{t}{2}$ for any $\lambda \in \mathbb{F}_q$,
and $\lambda_1$ is a most frequently occurring element in the set $\mathcal{B}$.
\end{proof}
\end{Theorem}

\subsection{Decoding algorithms for TGRS codes}
In this section, we will give decoding algorithms for TGRS codes $C_1$ and $C_2$
based on the extended Euclid's Algorithm.

\begin{Theorem}
    Let $C$ be a TGRS $[n, n-t, d]$ code as given in Definition~\ref{deftgrs1}, where $d=t$ or ($d=t+1$ and $t$ is odd).
Let $\boldsymbol{r}$ be a received word with $d(\boldsymbol{r}, C)\le \lfloor \frac{d-1}{2} \rfloor$,
$S(x)$ the syndrome polynomial of $\boldsymbol{r}$ as given in Equation (\ref{ss0}), and $g(x)= x^{t-1}.$
Let $\sigma _{i}( x)$ and $\tau _{i}( x)$ for $i \in \{ - 1, 0, \ldots , v + 1\}$ be the polynomials from the Euclid's Algorithm applied to $g(x)$ and $S(x)$.
Let $h$ be the minimum index such that $\deg \tau_{h}( x) < \lfloor \frac{t-1}{2} \rfloor.$
Then $(\sigma_{h}(x), \tau_{h}(x))$ satisfies Equations (\ref{eq0})-(\ref{eq1}).
Moreover, we can use Algorithm~\ref{alg1} to locate the error word $\boldsymbol{e}$.
\end{Theorem}
\begin{algorithm}
    \SetAlgoLined
    \SetKwInOut{Input}{input}\SetKwInOut{Output}{output}
    \Input{$\boldsymbol{r}:=(r_{1},r_{2},\ldots,r_{n})\in\mathbb{F}_{q}^{n}$.}
    \Output{$\boldsymbol{c}:=(c_{1},c_{2},\ldots,c_{n})\in\mathbb{F}_{q}^{n}$.}

    $\boldsymbol{s}=H_{1}\boldsymbol{r}^{T}=(s_{0},\ldots,s_{t-2})^{T}$, $S(x)=\sum_{i=0}^{t-2}s_{i}x^{i}$\;
    $\tau_{-1}(x)=g(x)$, $\tau_{0}(x)=S(x)$, $\sigma_{-1}(x)=0$, $\sigma_{0}(x)=1$, $h = -2$\;
    \Repeat{$\deg\tau_{h+2}(x)<\frac{t}{2}$}{
        $h = h+1$, $q_{h+2}(x) = \tau_h(x)$ div $\tau_{h+1}$\;
        $\tau_{h+2} = \tau_h$ mod $\tau_{h+1}$, $\sigma_{h+2} = \sigma_h- q_h\cdot\sigma_{h+1}$\;
    }
    $\sigma(x)=\sigma_{h+2}(x)$, $\tau(x)=\tau_{h+2}(x)$\;
    \For{$i=1,...,n$ }{
        $
            e_{i}= \begin{cases}
                -\frac{\alpha_i\tau(\alpha_{i}^{-1})}{w_{i}\sigma^{\prime}(\alpha_{i}^{-1})}, & \text{if} \ \sigma(\alpha_i^{-1}) = 0, \\
                0,  & \text{otherwise}.
            \end{cases}
        $
    }
    Output $\boldsymbol{e}=(e_{1},e_{2},\ldots,e_{n})$ and $\boldsymbol{c}=\boldsymbol{r}-\boldsymbol{e}.$
    \caption{$\lfloor \frac{d-1}{2} \rfloor$ Error-Correcting Decoding Algorithm for TGRS Codes}
    \label{alg1}
\end{algorithm}

\begin{Remark}
    In fact, Algorithm~\ref{alg1} is capable of correcting errors in twisted Goppa codes defined in \cite{Sui3}.
    Compared to the error correction algorithm presented in \cite{Sui3},
    it possesses the same error detection and correction capabilities but exhibits superior performance,
    as we have omitted some unnecessary calculations.
    More specifically, during the decoding process, the algorithm in \cite{Sui3} uses the matrix $H$ in (\ref{parh}),
    while Algorithm~\ref{alg1} uses the submatrix $H_1$ in (\ref{submala}).
    This feature can save some computational effort during the decoding process.
\end{Remark}

\begin{Theorem}
    Let $C$ be an MDS TGRS $[ n, n- t, t+ 1]$ code as given in Definition~\ref{deftgrs1}, where $t$ is even.
Let $\boldsymbol{r}$ be a received word with $d(\boldsymbol{r},C)\le \frac{t}{2}$,
$S(x)$ the syndrome polynomial of $\boldsymbol{r}$ as given in Equation (\ref{ss1}), and $g(x)= x^{t}.$
Let $\sigma_{i}( x)$ and $\tau_{i}(x)$ for $i \in \{ - 1, 0, \ldots , v + 1\}$ be the polynomials from the Euclid's Algorithm applied to $g(x)$ and $S(x)$.
Let $h$ be the minimum index such that $\deg \tau _{h}( x) < \frac{t}{2}.$

\noindent (1) If $\deg \sigma_{h}(x) < \frac{t}{2}$, then $( \sigma_{h}(x), \tau_{h}(x))$ satisfies  Equations (\ref{sxx1})-(\ref{sxx11}) and $d(\boldsymbol{r}, C) < \frac{t}{2}$.

\noindent (2) If $\deg \sigma_{h}(x) = \frac{t}{2}$, then there exists $\lambda \in \mathbb{F}_{q}$ such that $(\lambda \sigma_{h- 1}(x) + \sigma_{h}(x), \lambda \tau_{h-1}(x) + \tau_{h}(x))$
satisfies Equations (\ref{sxx1})-(\ref{sxx2}), $d(\boldsymbol{r}, C) = \frac{t}{2}$,
and $\lambda$ is one of the most frequent elements in the set $\mathcal{B}$, where
\[
    \mathcal{B}=\{\beta_i | i \le i \le n \}\backslash \infty \ \text{and} \ \beta_i = \begin{cases}
        \sigma_{h-1}(\alpha_i)^{-1} \sigma_{h}(\alpha_i), \ & \text{if} \ \sigma_{h-1}(\alpha_i) \ne 0, \\
        \infty,  & \text{if} \ \sigma_{h-1}(\alpha_i) = 0.
    \end{cases}
\]
Moreover, we can use Algorithm~\ref{alg2} to locate the error word $\boldsymbol{e}$.
\end{Theorem}

\begin{algorithm}
    \SetKwInOut{Input}{input}\SetKwInOut{Output}{output}
    \Input{$\boldsymbol{r}:=(r_{1},r_{2},\ldots,r_{n})\in\mathbb{F}_{q}^{n}$.}
    \Output{$\boldsymbol{c}:=(c_{1},c_{2},\ldots,c_{n})\in\mathbb{F}_{q}^{n}$.}
    $\boldsymbol{s}=H_{1}\boldsymbol{r}^{T}=(s_{0},\ldots,s_{t-1})^{T}, S(x)=\sum_{i=0}^{t-1-i}s_{i}x^{i}$\;
    $\tau_{-1}(x)=g(x), \tau_{0}(x)=S(x)$, $\sigma_{-1}(x)=0$, $\sigma_{0}(x)=1$, $h = -2$\;
    \Repeat{$\deg\tau_{h+2}(x)<\frac{t}{2}$}{
        $h = h+1$, $q_{h+2}(x) = \tau_h(x)$ div $\tau_{h+1}$\;
        $\tau_{h+2} = \tau_h$ mod $\tau_{h+1}$, $\sigma_{h+2} = \sigma_h- q_h\cdot\sigma_{h+1}$\;
    }
    \eIf{$\deg \sigma_{h+2}(x)<\frac{t}{2}$}{
        $\sigma(x)=\sigma_{h+2}(x)$, $\tau(x)=\tau_{h+2}(x)$\;
        \For{$i=1,...,n$ }{
            $
                e_{i}= \begin{cases}
                    -\frac{\tau(\alpha_{i})}{w_{i}\alpha_{i}^t\sigma^{\prime}(\alpha_{i})}, & \text{if} \ \sigma(\alpha_i) = 0, \\
                    0,  & \text{otherwise}.
                \end{cases}
            $
        }
    }{
        \For{$i=1,...,n$ }{
            $
                \beta_i = \begin{cases}
                    \sigma_{h+1}(\alpha_i)^{-1} \sigma_{h+2}(\alpha_i), \ & \text{if} \ \sigma_{h+1}(\alpha_i) \ne 0, \\
                    \infty,  & \text{if} \ \sigma_{h+1}(\alpha_i) = 0.
                \end{cases}
            $
        }
        \For{$\lambda$ in $FrequentEle(\{\beta_i\})$}{
            \tcp{Obtain all most frequent elements of set $\{\beta_i\}$ with $\infty$ excluded.}
            $\sigma(x)=\lambda\sigma_{h+1}(x)+\sigma_{h+2}(x)$, $\tau(x)=\lambda\tau_{h+1}(x)+\tau_{h+2}(x)$\;
            \For{$i=1,...,n$ }{
            $
                e_{i}= \begin{cases}
                    -\frac{\tau(\alpha_{i})}{w_{i}\alpha_{i}^t\sigma^{\prime}(\alpha_{i})}, & \text{if} \ \sigma(\alpha_i) = 0, \\
                    0,  & \text{otherwise}.
                \end{cases}
            $
            }
        \If{
            $\tau(x)=a \sigma(x)+\omega(x), a=\sum_{j \in J} e_jw_j f(\alpha_j) \in \mathbb{F}_q, \omega(x)=-\sum_{j\in J}e_jw_j\alpha_j^t\frac{\sigma(x)}{x-\alpha_j}$
        }{
            break\;
        }
        }
    }
    Output $\boldsymbol{e}=(e_{1},e_{2},\ldots,e_{n})$ and $\boldsymbol{c}=\boldsymbol{r}-\boldsymbol{e}.$
\caption{$\lfloor \frac{t}{2} \rfloor$ Error-Correcting Decoding Algorithm for TGRS Codes}
\label{alg2}
\end{algorithm}

\begin{Remark}
The decoding algorithm for TGRS codes in \cite{Sun1} employed an exhaustive search of $\lambda \in \mathbb{F}_q$ to determine the polynomial pair $(\sigma(x), \tau(x))$ when decoding up to $\frac{t}{2}$ errors.
In contrast, when decoding TGRS codes with up to $\frac{t}{2}$ errors
using the approach outlined in Theorem \ref{newthm}, we can search for $\lambda$ within a smaller,
more restricted range $\mathcal{B}$ (see (\ref{betadef})) to determine the polynomial pair $(\sigma(x), \tau(x))$.
This results in our decoding algorithm having better performance, 
Detailed comparison results can be found in the conclusion of this paper.
\end{Remark}

In the following, we use an example to demonstrate the decoding process of Algorithm~\ref{alg2}.
\begin{Example}
    Let $\mathbb{F}_{2^6}=\mathbb{F}_2\langle z \rangle$ with $z^6 + z^4 + z^3 + z + 1 = 0$.
Let $\boldsymbol{\alpha} = (\alpha_1, . . . , \alpha_8) = ( z^{33}, z^{56}, z^{47}, z^{3}, z^{25}, z^{50}, z^{20}, z^{32} )$,
$\boldsymbol{v} = (v_1, . . . , v_8) = (z^{56}, z^{45}, z^{28}, z^{59}, z^{60}, z^{25}, z^{53}, z^{13})$ and $\eta = z^{39}$.
Let $C_3=\TGRS_{4,4}(\boldsymbol{\alpha}, \boldsymbol{v}, 2, z^{39}, \boldsymbol{1})$ be an MDS TGRS code over $\mathbb{F}_{2^6}$ with a generator matrix
$$G_3=\begin{pmatrix}
    v_1&\cdots&v_8  \\
    v_1(\alpha_1+\eta\alpha_1^4)&\cdots&v_8(\alpha_8+\eta\alpha_8^4) \\
    v_1\alpha_1^2&\cdots&v_8\alpha_8^2 \\
    v_1\alpha_1^3&\cdots&v_8\alpha_8^3
\end{pmatrix}=
\begin{pmatrix}
    z^{56}& z^{45} & z^{28} &z^{59} & z^{60} & z^{25} & z^{53} &z^{13}\\
z^{15} & z^{29} & z^{30} &z^{18} & z^{62}  &  0 &z^{55} &  z^{9}\\
z^{59} & z^{31} & z^{59} & z^{2} & z^{47} & z^{62} &z^{30} &z^{14}\\
z^{29} & z^{24} & z^{43} &z^{5} & z^{9} &z^{49} & z^{50} &z^{46}
\end{pmatrix},$$ and a parity-check matrix
$$H_3=\begin{pmatrix}
    w_1&\cdots&w_8  \\
    w_1\alpha_1&\cdots&w_8\alpha_8 \\
    w_1\alpha_1^2&\cdots&w_8\alpha_8^2 \\
    w_1(\alpha_1^3 + f(\alpha_1))&\cdots&w_8 (\alpha_8^3 + f(\alpha_8))
\end{pmatrix}=
\begin{pmatrix}
    z^{6} &z^{53} &z^{32} &z^{24} &z^{42} &z^{13} &z^{19} &z^{26} \\
    z^{39} &z^{46} &z^{16} &z^{27}  &z^{4}    &1 &z^{39} &z^{58} \\
    z^{9} &z^{39} &   1 &z^{30} &z^{29} &z^{50} &z^{59} &z^{27} \\
    z^{39} &z^{52} &z^{33} &z^{15} &z^{49} &z^{13} &z^{47} &z^{62}
\end{pmatrix},$$
where $f=x^6 + z^{44}x^5 + z^{19}x^4 + x^3$. Assume that $\boldsymbol{c} = (z^{9}, z^{25}, z^{56}, z^{26}, z^{45}, z^{59}, z^{19}, z^{13})$ and $\boldsymbol{e} = (0, 0, z^{7}, 0, 0, 0, z^{36}, 0)$.
Then the received word is $\boldsymbol{r} = \boldsymbol{c} + \boldsymbol{e} = ( z^{9}, z^{25}, z^{9}, z^{26}, z^{45}, z^{59}, z^{58}, z^{13})$.
Input $\boldsymbol{r}$ to Algorithm~\ref{alg2}. Then $\boldsymbol{s} = (z^{53}, z^{35}, z^2,z^{14})^T$ and $S(x) =z^{53}x^3 + z^{35}x^2 + z^{2}x + z^{14}$.
Applying the Euclid's Algorithm to $x^4$ and $S(x)$,
we have Table 1.
\begin{table}[H]
    \fontsize{10pt}{6pt}
    \caption{The Euclid's Algorithm process}
        \centering
        \renewcommand\arraystretch{0.9}
        \begin{tabular}{cccc} \hline
        $j$ & $q_j(x)$ &  $\sigma_j(x)$ & $\tau_j(x)$  \\ \hline
        $-1$ &   &  $0$  & $x^4$   \\
        $0$ &   &  $1$  & $S(x)$  \\
        $1$ & $z^{10}x + z^{55}$  &  $z^{10}x + z^{55}$  & $z^{46}x^2 + z^{62}x + z^{6}$ \\
        $2$ & $z^{7}x + z^{4}$  &  $z^{17}x^2 + z^{33}x + z^{31}$ & $z^{49}x + z^{45}$ \\ \hline
    \end{tabular}
\end{table}
Here $h = 2$ is the minimum index such that  $\deg \sigma_h(x) = \frac{t}{2} = 2$
and $\deg \tau_h(x) < 2$. Then
the set $\{\beta_j\}=\{z^{22}, z^{38}, z^{26}, z^{22}, z^{20}, z^{44}, z^{26}, z^{5}\}$ has the most frequent elements $z^{22}$ and $z^{26}$.

Set
$(\sigma(x), \tau (x)) = (z^{22}\sigma_1(x) + \sigma_2(x), z^{22}\tau_1(x) + \tau_2(x))$.
Then $\sigma(x) = z^{17}x^2 + z^{25}x + z^{53}$ and $\tau (x) = z^5x^2 + z^{53}x + z^4$.
Following the calculation, $J = \{1, 4\}$, $\boldsymbol{e} = ( 1, 0, 0, z^{43}, 0, 0, 0, 0 )$, and $a=\sum_{i \in J}  e_iw_i f(\alpha_i)=z^{62}$.
It is easy to verify that $\tau(x)$ is not equal to $a\sigma(x) + \omega(x)$. Thus it can be eliminated.

Next, set
$(\sigma(x), \tau (x)) = (z^{26}\sigma_1(x) + \sigma_2(x), z^{26}\tau_1(x) + \tau_2(x))$.
Then $\sigma(x) = z^{17}x^2 + z^{46}x + z^{21}$ and $\tau (x) = z^{9}x^2 + z^{3}x + z^{35}$.
Following the calculation, $J = \{1, 4\}$, $\boldsymbol{e} = ( 0, 0, z^7, 0, 0, 0, z^{36}, 0  )$, and $a=\sum_{i \in J}  e_iw_i f(\alpha_i)=z^{55}$.
After verification, we can get that $\tau(x)$ is equal to $a\sigma(x) + \omega(x)$.

Finally, the output $\boldsymbol{c} = \boldsymbol{r} - \boldsymbol{e} = (z^{9}, z^{25}, z^{56}, z^{26}, z^{45}, z^{59}, z^{19}, z^{13})$.
\end{Example}

Here, we present a very specific example to demonstrate that there can be multiple elements in the set $\mathcal{B}$ (see (\ref{betadef})) with the highest frequency of occurrence.
In the above example, there are two such elements.
In fact, through computations and observations, we have found that in most cases,
there is only one element in the set $\mathcal{B}$ that appears most frequently.

\section{Twisted Goppa Codes}
Classical Goppa codes were introduced by Goppa in 1970 (\cite{Goppa1, Goppa2}).
Goppa codes are subfield subcodes of a class of GRS codes.
Similarly, twisted Goppa codes are subfield subcodes of a class of TGRS codes (\cite{Sui3, Sun1}).
In this section, we extend the definitions of twisted Goppa codes.
The decoding algorithms for TGRS codes that we provided above can be applied to the Goppa codes defined as follows.

Let $q=p^m$, where $p$ is a prime and $m$ is a positive integer.

\begin{Definition}
    Let $g(x)$ be a monic polynomial of degree t over $\mathbb{F}_{p^m}$,
$\mathcal{L}=\{\alpha_i\,|\, 1\le i \le n\} \subseteq \mathbb{F}_{p^m}$ a defining set such that $g(\alpha_i) \neq 0$
for all $\alpha_i \in \mathcal{L}$, and $f(x) \in \mathbb{F}_{p^m}[x]$.
Then a twisted Goppa code over $\mathbb{F}_p$
with respect to $\mathcal{L}$, $g(x)$ and $f(x)$ is defined as
\[
    \Gamma(\mathcal{L}, g, f)=\left\{c=(c_1,...,c_n) \in \mathbb{F}_p^n\,|\, \sum_{i=1}^nc_i\left(\frac{1}{x-\alpha_i} - \frac{f(\alpha_i)}{g(\alpha_i)}\right)\equiv 0 \pmod{g(x)}\right\}.
\]
Note that if $f(x)=0$, then $\Gamma(\mathcal{L}, g, f)$ is the Goppa code.
\end{Definition}

\begin{Proposition}
Assume the notation is as given above. Then
    $$\Gamma(\mathcal{L}, g, f)=\{\boldsymbol{c}=(c_1,...,c_n) \in \mathbb{F}_p^n\,|\, H\boldsymbol{c}^T = 0\},$$
where
\begin{align}\label{goppach}
    H=\begin{pmatrix}
        \frac{1}{g(\alpha_1)}  & \cdots & \frac{1}{g(\alpha_n)} \\
        \frac{1}{g(\alpha_1)} \alpha_1  & \cdots & \frac{1}{g(\alpha_n)} \alpha_n \\
        \vdots &  \ddots & \vdots \\
        \frac{1}{g(\alpha_1)} \alpha_1^{t-2}  & \cdots & \frac{1}{g(\alpha_n)} \alpha_n^{t-2} \\
        \frac{1}{g(\alpha_1)} (\alpha_1^{t-1} + f(\alpha_1))  & \cdots & \frac{1}{g(\alpha_n)} (\alpha_n^{t-1} +f(\alpha_n))
    \end{pmatrix}.
\end{align}
\begin{proof}

Let $g(x)=\sum_{j=0}^tg_jx^j \in \mathbb{F}_{p^m}[x]$ with $g_t=1$.
Then in the quotient ring $\mathbb{F}_{p^m}[x]/(g(x))$,
\begin{align*}
    \frac{1}{x-\alpha_i}-\frac{f(\alpha_i)}{g(\alpha_i)} & = -\frac{1}{g(\alpha_i)}\left(\frac{g(x)-g(\alpha_i)}{x-\alpha_i} + f(\alpha_i)\right) \\
    & = -\frac{1}{g(\alpha_i)}\left(\sum_{j=1}^tg_j\sum_{l=0}^{j-1}x^l\alpha_i^{j-l-1} + f(\alpha_i)\right) \\
    & = -\frac{1}{g(\alpha_i)}\left(\sum_{l=0}^{t-1}x^l\sum_{j=l+1}^tg_j\alpha_i^{j-l-1} + f(\alpha_i)\right).
\end{align*}
So, by the definition of twisted Goppa code, $\boldsymbol{c}=(c_1,...,c_n) \in \Gamma(\mathcal{L}, g, f)$ if and only if
\[
   \sum_{i=1}^n\frac{1}{g(\alpha_i)}\left(\sum_{l=0}^{t-1}x^l\sum_{j=l+1}^tg_j\alpha_i^{j-l-1} + f(\alpha_i)\right)c_i \equiv 0 \pmod{g(x)}.
\]
Therefore, setting the coefficients of $x^l$ equal to $0$, in the order $l=t-1,t-2,...,0$,
we have that $\boldsymbol{c} \in \Gamma(\mathcal{L}, g, f)$ if and only if $H' \boldsymbol{c}^T=\boldsymbol{0}$, where
\[
    H'=\begin{pmatrix}
        \frac{1}{g(\alpha_1)}  & \cdots & \frac{1}{g(\alpha_n)} \\
        \frac{1}{g(\alpha_1)} \sum_{i=t-1}^tg_i\alpha_1^{i-t+1}  & \cdots & \frac{1}{g(\alpha_n)} \sum_{i=t-1}^tg_i\alpha_n^{i-t+1} \\
        \vdots &  \ddots & \vdots \\
        \frac{1}{g(\alpha_1)} \sum_{i=2}^tg_i\alpha_1^{i-2}  & \cdots & \frac{1}{g(\alpha_n)}\sum_{i=2}^tg_i \alpha_n^{i-2} \\
        \frac{1}{g(\alpha_1)} (\sum_{i=1}^tg_i\alpha_1^{i-1} + f(\alpha_1))  & \cdots & \frac{1}{g(\alpha_n)} (\sum_{i=1}^tg_i\alpha_n^{i-1} +f(\alpha_n))
    \end{pmatrix}.
\]
Here, $H'$ can be row reduced to the $t \times n$ matrix in (\ref{goppach}).
\end{proof}
\end{Proposition}

\begin{Remark}
    When $\boldsymbol{w}=(w_1,...,w_n)$ is taken as $(\frac{1}{g(\alpha_1)},...,\frac{1}{g(\alpha_n)})$,
    the code $\Gamma(\mathcal{L}, g, f)$ has a parity-check matrix in the form $H$ given in (\ref{parh}).
    Therefore, $\Gamma(\mathcal{L}, g, f)$ is a subfield subcode of TGRS code $C$ mentioned in the beginning of Section 4.
\end{Remark}

Based on the relationship between a code and its subfield subcode, we can easily draw the following conclusion.
\begin{Proposition}
    Let $\Gamma(\mathcal{L}, g, f)$ be an $[n,k,d]$ linear code over $\mathbb{F}_{p}$.
Then

(1) $d \ge t+1$, if the code with the parity check matrix (\ref{goppach}) is MDS,

(2) $d \ge t$, if the code with parity the check matrix (\ref{goppach}) is almost-MDS,

\noindent and $k \ge n-mt$, where $t$ denotes the degree of the polynomial $g(x)$.
\end{Proposition}

When performing $\lfloor \frac{t-1}{2} \rfloor$ or $\lfloor \frac{t}{2} \rfloor$ error-correction decoding on the $[n,k,d]$ $\Gamma(\mathcal{L}, g, f)$ code,
we can utilize the previously discussed theoretical results and make slight modifications to Algorithms \ref{alg1} and \ref{alg2} for their application.
Therefore, we do not elaborate further on this point.

\section{Conclusions}

In this paper, we studied the decoding of a more general class of twisted generalized Reed-Solomon codes and provided a more precise characterization of the key equation for TGRS codes.
This characterization aided in optimizing the algorithm presented in \cite{Sun1}, and we also proposed the optimized decoding algorithm.
We further studied the decoding of almost-MDS TGRS codes and provided the optimized decoding algorithm which is more efficient than the decoding algorithm in \cite{Sui3} in performance.
The optimized decoding algorithms can be applied to the decoding of a more general class of twisted Goppa codes.

The following table compares the decoding times between Algorithm~\ref{alg2} in this paper and Algorithm~2 in \cite{Sun1}.
For each parameter of TGRS codes, two samples were selected, and the decoding algorithm was repeatedly performed 10,000 times to record the time consumption (Units: seconds).
During each decoding run, $\lfloor \frac{d-1}{2} \rfloor$ new random errors were generated.
For the convenience of our comparative testing, we made partial adjustments to Algorithm 2 in \cite{Sun1} so that it could be applied to the TGRS codes defined in this paper.
All computations were performed on a Windows 10 system with an Intel Core i3-10100 processor using Magma \cite{Bosma} (version 2.25-3).
\footnote{The Magma code can be found in https://github.com/1wangguodong/Decoding-twisted-generalized-Reed-Solomon-Codes}

\begin{table}[H]
    \caption{Performance comparison}
    \begin{center}
    { 
        \renewcommand\arraystretch{0.88}
        \begin{tabular}{cccccccc|ccccccccc} \hline
        $n$ & $k$ &  $d$ & $r$ & $t_1$ & $t_2$ & $t_1'$ & $t_2'$ & $n$ & $k$ &  $d$ & $r$ & $t_1$ & $t_2$ & $t_1'$ & $t_2'$   \\ \hline \hline
         13   &  9  & 5  & 1 & 17.532 & 1.281 & 15.500 & 1.313 & 11   &  5  & 7  & 1 &  16.437 & 1.453 & 17.985 & 1.468 \\
         13   &  9  & 5  & 2 & 16.219 & 1.375 & 16.016 & 1.406 & 11   &  5  & 7  & 2 & 17.218 & 1.407 & 17.297 & 1.438 \\
         13   &  9  & 5  & 3 & 16.531 & 1.407 & 15.343 & 1.328 & 11   &  5  & 7  & 3 & 17.625 & 1.391 & 14.297 & 1.219 \\
         13   &  9  & 5  & 4 & 17.532 & 1.453 & 15.235 & 1.312 & 11   &  5  & 7  & 4 & 17.641 & 1.687 & 15.015 & 1.719 \\
         13   &  9  & 5  & 5 & 16.171 & 1.421 & 15.031 & 1.297 & 11   &  5  & 7  & 5 & 17.562 & 1.594 & 14.000 & 1.172 \\
         13   &  9  & 5  & 6 & 16.891 & 1.375 & 15.313 & 1.328 & 10   &  6  & 5   &  1 & 12.859 & 1.078 & 14.609 & 1.282 \\
         13   &  9  & 5  & 7 & 17.140 & 1.391 & 15.250 & 1.313 & 10   &  6  & 5   &  2 & 14.797 & 1.203 & 13.953 & 1.578 \\
         13   &  9  & 5  & 8 & 16.485 & 1.406 & 15.281 & 1.265 & 10   &  6  & 5   &  3 & 16.265 & 1.297 & 16.281 & 1.344 \\
         13   &  9  & 5  & 9 & 16.875 & 1.359 & 15.062 & 1.250 &  10   &  6  & 5   &  4 & 17.078 & 1.156 & 15.734 & 1.359 \\

         12   &  6  & 7  & 1 & 18.344 & 1.765 & 14.594 & 1.219 & 10   &  6  & 5   &  5 & 17.219 & 1.172 & 15.516 & 1.234 \\
         12   &  6  & 7  & 2 & 19.094 & 1.781 & 14.609 & 1.500 &  10   &  6  & 5   &  6 & 13.656 & 1.188 & 15.829 & 1.406 \\
         12   &  6  & 7  & 3 & 19.469 & 1.547 & 16.953 & 1.625 \\
         12   &  6  & 7  & 4 & 18.015 & 1.797 & 17.609 & 1.594 \\
         12   &  6  & 7  & 5 & 17.360 & 1.671 & 16.375 & 1.391 \\
         12   &  6  & 7  & 6 & 19.453 & 1.625 & 16.719 & 1.312 \\
    \hline
    \end{tabular}}
    \end{center}
    \begin{threeparttable}
        \begin{tablenotes}
            \footnotesize
            \item[1]
            In this table, the parameters `$n$, $k$, $d$, $r$' denote the code length, dimension, minimum distance, and the twisted row, respectively.
            The symbols `$t_1$' and `$t_1'$' denote the execution times of Algorithm 2 from \cite{Sun1}, whereas `$t_2$' and `$t_2'$' denote the execution times of Algorithm \ref{alg2} in this paper.
        \end{tablenotes}
    \end{threeparttable}
\end{table}
\vskip 2mm
\noindent\textbf{Acknowledgement.}

This work was supported by the National Natural Science Foundation of China (Grant Nos. 12271199, 12441102, 12171191).

\end{document}